\documentclass[11pt]{article}
\usepackage[utf8]{inputenc}

\usepackage{fullpage}

\usepackage{verbatim}

\usepackage{graphicx,color}
\usepackage{tikz}

\usepackage{amssymb} 
\usepackage{amsmath} 
\usepackage{amsthm} 
\usepackage{setspace} 
\usepackage{cite} 
\usepackage{hyperref}
\usepackage{tikz}
\usepackage{verbatim}
\usetikzlibrary{shapes.geometric}
\usepackage{graphicx}
\usepackage{caption}
\usepackage{subcaption}

\newtheorem{definition}{Definition}

\newtheorem{claim}{Claim} 
\newtheorem{lemma}{Lemma} 
\newtheorem{theorem}{Theorem} 
 
\newtheorem{obs}{Observation} 

\newtheorem{proposition}{Proposition}
\newtheorem{remark}{Remark}

\newcommand{\ex}{\mbox{\small\rm X}}

\newcommand{\F}{\mathbb{F}}

\newcommand{\B}{\mathbb{B}}

\newcommand{\Q}{\mathbb{Q}}
\newcommand{\R}{\mathbb{R}}

\renewcommand{\angle}[1]{\mathopen{\langle} #1\mathclose{\rangle}}

\newcommand{\pit}{\mbox{\small\rm PIT}}

\DeclareMathOperator{\poly}{\mbox{\small\rm poly}}

\title{Efficient Identity Testing and Polynomial Factorization over Non-associative Free Rings}

\author{V. Arvind\thanks{Institute of Mathematical Sciences (HBNI), Chennai,
    India, \texttt{email: arvind@imsc.res.in}} \and Rajit Datta
  \thanks{Chennai Mathematical Institute, Chennai, India,
    \texttt{email: rajit@cmi.ac.in}} \and Partha Mukhopadhyay \thanks{Chennai Mathematical Institute, Chennai, India,
    \texttt{email: partham@cmi.ac.in}} \and S. Raja\thanks{Chennai Mathematical Institute, Chennai, India,
    \texttt{email: sraja@cmi.ac.in}}}

\begin{document}

\maketitle

\begin{abstract}
  In this paper we study arithmetic computations in the nonassociative, and
  noncommutative free polynomial ring $\F\{x_1,x_2,\ldots,x_n\}$. Prior to this
  work, nonassociative arithmetic computation was considered by Hrubes,
  Wigderson, and Yehudayoff \cite{HWY10b}, and they showed lower bounds and
  proved completeness results. We consider Polynomial Identity Testing (\pit)
  and polynomial factorization over $\F\{x_1,x_2,\ldots,x_n\}$ and show the
  following results.

\begin{enumerate}

\item Given an arithmetic circuit $C$ of size $s$ computing a polynomial
  $f\in\F\{x_1,x_2,\ldots,x_n\}$ of degree $d$, we give a deterministic
  $\poly(n,s,d)$ algorithm to decide if $f$ is identically zero polynomial or
  not. Our result is obtained by a suitable adaptation of the PIT algorithm of
  Raz-Shpilka\cite{RS05} for noncommutative ABPs.

\item Given an arithmetic circuit $C$ of size $s$ computing a
  polynomial $f\in\F\{x_1,x_2,\ldots,x_n\}$ of degree $d$, we give an
  efficient deterministic algorithm to compute circuits for the
  irreducible factors of $f$ in time $\poly(n,s,d)$ when $\F=\Q$. Over
  finite fields of characteristic $p$, our algorithm runs in time
  $\poly(n,s,d,p)$.
\end{enumerate}
\end{abstract}

\section{Introduction}\label{intro}

Noncommutative computation, introduced in complexity theory by Hyafil
\cite{Hya77} and Nisan \cite{N91}, is an important subfield of algebraic
complexity theory. The main algebraic structure of interest is the free
noncommutative ring $\F\angle{X}$ over a field $\F$, where
$X=\{x_1,x_2,\cdots,x_n\}$ is a set of free noncommuting variables. A central
problem is Polynomial Identity Testing which may be stated as follows:

Let $f\in \F \angle{X}$ be a polynomial represented by a noncommutative
arithmetic circuit $C$.  The circuit $C$ can either be given by a black box
(using which we can evaluate $C$ on matrices with entries from $\F$ or an
extension field), or the circuit may be explicitly given. The algorithmic
problem is to check if the polynomial computed by $C$ is identically zero.
We recall the formal definition of a noncommutative arithmetic circuit.

\begin{definition}\label{ckt}
  An \emph{arithmetic circuit} $C$ over a field $\F$ and
  indeterminates $X=\{x_1,x_2,\cdots,x_n\}$ is a directed acyclic
  graph (DAG) with each node of indegree zero labeled by a variable or
  a scalar constant from $\F$: the indegree $0$ nodes are the input
  nodes of the circuit. Each internal node of the DAG is of indegree
  two and is labeled by either a $+$ or a $\times$ (indicating that it
  is a plus gate or multiply gate, respectively). Furthermore, the two
  inputs to each $\times$ gate are designated as left and right inputs
  which prescribes the order of multiplication at that gate. A gate of
  $C$ is designated as \emph{output}. Each internal gate computes a
  polynomial (by adding or multiplying its input polynomials), where
  the polynomial computed at an input node is just its label. The
  \emph{polynomial computed} by the circuit is the polynomial computed
  at its output gate.
\end{definition}

When the multiplication operation of the circuit in
Definition~\ref{ckt} is \emph{noncommutative}, it is called a
\emph{noncommutative arithmetic circuit} and it computes a polynomial
in the free noncommutative ring $\F\angle{X}$.  Since cancellation of
terms is restricted by noncommutativity, intuitively it appears
noncommutative polynomial identity testing would be easier than
polynomial identity testing in the commutative case. This intuition is
supported by fact that there is a deterministic polynomial-time
white-box $\pit$ algorithm for noncommutative ABP \cite{RS05}. In the
commutative setting a deterministic polynomial-time PIT for ABPs would
be a major breakthrough.\footnote{The situation is similar even in the
  lower bound case where Nisan proved that noncommutative determinant
  or permanent polynomial would require exponential-size algebraic
  branching program \cite{N91}.}  However, there is little progress
towards obtaining an efficient deterministic $\pit$ for general
noncommutative arithmetic circuits. For example, the problem is open
even for noncommutative \emph{skew} circuits.

If \emph{associativity} is also dropped then it turns out that $\pit$
becomes easy, as we show in this work. More precisely, we consider the
free noncommutative and nonassociative ring of polynomials $\F\{X\}$,
$X=\{x_1,x_2,\ldots,x_n\}$, where a polynomial is an $\F$-linear
combination of monomials, and each monomial comes with a bracketing
order of multiplication. For example, in the nonassociative ring
$\F\{X\}$ the monomial $(x_1(x_2 x_1))$ is different from monomial
$((x_1 x_2)x_1)$, although in the associative ring $\F\angle{X}$ they
clearly coincide. 

When the multiplication operation is both noncommutative and
nonassociative, it is called a \emph{nonassociative noncommutative
 circuit} and it computes a polynomial in the free nonassociative
noncommutative ring $\F\{X\}$. Previously, the nonassociative
arithmetic model of computation was considered by Hrubes, Wigderson,
and Yehudayoff \cite{HWY10b}. They showed completeness and explicit
lower bound results for this model.  We show the following result
about $\pit$.

\begin{itemize}
\item  Let $f(x_1,x_2,\ldots,x_n)\in\F\{X\}$ be a degree $d$ polynomial
  given by an arithmetic circuit of size $s$. Then in deterministic
  $\poly(s,n,d)$ time we can test if $f$ is an identically zero
  polynomial in $\F\{X\}$.
\end{itemize}

\begin{remark}
  We note that our algorithm in the above result does not depend on
  the choice of the field $\F$. A recent result of Lagarde et
  al.\ \cite{LMP16} shows an exponential lower bound, and a
  deterministic polynomial-time $\pit$ algorithm over $\R$ for
  noncommutative circuits where all parse trees in the circuit are
  isomorphic. We also note that in \cite{AR16} an exponential lower
  bound is shown for set-multilinear arithmetic circuits with the
  additional semantic constraint that each monomial has a unique parse
  tree in the circuit (but different monomials can have different
  parse trees).
\end{remark}

Next, we consider polynomial factorization in the ring
$\F\{\ex\}$. Polynomial factorization is very well-studied in the
commutative ring $\F[X]$: Given an arithmetic circuit $C$ computing a
multivariate polynomial $f\in\F[\ex]$ of degree $d$, the problem is to
efficiently compute circuits for the irreducible factors of $f$. A
celebrated result of Kaltofen \cite{Kal89} solves the problem in
randomized $\poly(n,s,d)$ time. Whether there is a polynomial-time
deterministic algorithm is an outstanding open problem. Recently, it
is shown (for fields of small characteristic and characteristic zero)
that the complexity of deterministic polynomial factorization problem
and the $\pit$ problem are polynomially equivalent \cite{KSS15}. A
natural question is to determine the complexity of polynomial
factorization in the noncommutative ring $\F\angle{X}$. The free
noncommutative ring $\F\angle{X}$ is not even a \emph{unique
  factorization domain} \cite{Cohn}. However, unique factorization
holds for homogeneous polynomials in $\F\angle{X}$, and it is shown in
\cite{AJR15} that for homogeneous polynomials given by noncommutative
circuits, the unique factorization into irreducible factors can be
computed in randomized polynomial time (essentially, by reduction to
the noncommutative PIT problem).

In this paper, we note that the ring $\F\{X\}$ is a \emph{unique
  factorization domain}, and given a polynomial in $\F\{X\}$ by a
circuit, we show that circuits for all its irreducible factors can be
computed in deterministic polynomial time.

\begin{itemize}
\item  Let $f(x_1,x_2,\ldots,x_n)\in\F\{X\}$ be a degree $d$ polynomial
  given by an arithmetic circuit of size $s$.  Then if $\F=\Q$, in
  deterministic $\poly(s,n,d)$ time we can output the circuits for the
  irreducible factors of $f$. If $\F$ is a finite field such that
  $char(\F)=p$, we obtain a deterministic $\poly(s,n,d,p)$ time
  algorithm for computing circuits for the irreducible factors of $f$.
\end{itemize}

\subsection*{Outline of the proofs}

\begin{itemize}
\item \textbf{Identity Testing Result:} The main ideas for our
  algorithm are based on the white-box Raz-Shpilka PIT algorithm for
  noncommutative ABPs \cite{RS05}. As in the Raz-Shpilka algorithm
  \cite{RS05}, if the circuit computes a nonzero polynomial $f \in
  \F\{X\}$, then our algorithm output a \emph{certificate monomial}
  $m$ such that coefficient of $m$ in $f$ is nonzero.


  We first sketch the main steps of the Raz-Shpilka algorithm. The
  Raz-Shpilka algorithm processes the input ABP (assumed homogeneous)
  layer by layer. Suppose layer $i$ of the ABP has $w$ nodes. The
  algorithm maintains a spanning set $\mathbb{B}_i$ of at most $w$
  many linearly independent $w$-dimensional vectors of monomial
  coefficients. More precisely, each vector $v_m\in\mathbb{B}_i$ is
  the vector of coefficients of monomial $m$ computed at each of the
  $w$ nodes in layer $i$. Furthermore, the coefficient vector at layer
  $i$ of any monomial is in the span of $\mathbb{B}_i$. The
  construction of $\mathbb{B}_{i+1}$ from $\mathbb{B}_i$ can be done
  efficiently.  Clearly the identity testing problem can be solved by
  checking if there is a nonzero vector in $\mathbb{B}_d$, where $d$
  is the total number of layers.

  Now we sketch our $\pit$ algorithm for polynomials over $\F\{\ex\}$
  given by circuits. Let $f$ be the input polynomial given by the
  circuit $C$.
  
  We encode monomials in the free nonassociative noncommutative ring
  $\F\{\ex\}$ as monomials in the free noncommutative ring
  $\F\angle{\ex,(,)}$, such that the encoding preserves the
  multiplication structure of $\F\{\ex\}$ (Observation \ref{obs1}).
  For $1\leq j\leq d$, we can efficiently find from $C$ a homogeneous
  circuit $C_j$ that computes the degree $j$ homogeneous part of $C$.
  Thus, it suffices to test if $C_j\equiv 0$ for each $j$. Hence, it
  suffices to consider the case when $f\in\F\{\ex\}$ is homogeneous
  and $C$ is a homogeneous circuit computing $f$.

  For $j\leq d$ let $G_j$ denote the set of degree $j$ gates of
  $C$. The algorithm maintains a set $\mathbb{B}_j$ of
  $|G_j|$-dimensional linearly independent vectors of monomial
  coefficients such that any degree $j$ monomial's coefficient vector
  is in the linear span of $\mathbb{B}_j$. Clearly, $|\mathbb{B}_j|\le
  |G_j|$. We compute $\mathbb{B}_{j+1}$ from the sets $\{\mathbb{B}_i
  : 1\leq i\leq j\}$. For each vector in $\mathbb{B}_j$ we also keep
  the corresponding monomial. In the nonassociative model a degree $d$
  monomial $m=(m_1 m_2)$ is generated in a \emph{unique way}. To check
  if the coefficient vector of $m$ is in the span of $\mathbb{B}_d$ it
  suffices to consider vectors in the spans of $\mathbb{B}_{d_1}$ and
  $\mathbb{B}_{d_2}$, where $d_1=\deg(m_1)$ and $d_2=\deg(m_2)$. This
  is a crucial difference from a general noncommutative circuit and
  using this property we can compute $\mathbb{B}_{j+1}$.

 \item \textbf{Polynomial Factorization in $\F\{X\}$}
 
   For a polynomial $f\in\F\{\ex\}$, let $f_j$ denote the homogeneous
   degree $j$ part of $f$. For a monomial $m$, let $c_m(f)$ denote the
   coefficient of $m$ in $f$. We will use the $\pit$ algorithm as
   subroutine for the factoring algorithm. Arvind et al.\ \cite{AJR15}
   have shown that given a monomial $m$ and a homogeneous
   noncommutative circuit $C$, in deterministic polynomial time
   circuits for the formal left and right derivatives of $C$ with
   respect to $m$ can be efficiently computed. This result is another
   ingredient in our algorithm.

   We sketch the easy case, when the given polynomial $f$ of degree
   $d$ has no constant term. Applying our $\pit$ algorithm to the
   homogeneous circuit $C_d$ (computing $f_d$) we find a nonzero
   monomial $m=(m_1 ~m_2)$ of degree $d$ in $f_d$ along with its
   coefficient $c_m(f)$. Notice that for any nontrivial factorization
   $f=gh$, $m_1$ is a nonzero monomial in $g$ and $m_2$ is a nonzero
   monomial in $h$. Suppose $|m_1|=d_1$ and $|m_2|=d_2$. Then the left
   derivative of $C_d$ with respect to $m_1$ gives
   $c_{m_1}(g)~h_{d_2}$ and the right derivative of $C_d$ with respect
   to $m_2$ gives $c_{m_2}(h)~g_{d_1}$. We now use the circuits for
   these derivatives and the nonassociative structure, to find
   circuits for different homogeneous parts of $g$ and $h$. The
   details, including the general case when $f$ has a nonzero constant
   term, is in Section~\ref{sec:fact}.
\end{itemize}

\subsection*{Organization}
In Section \ref{prelim} we describe some useful properties of
nonassociative and noncommutative polynomials. In Section
\ref{sec:pit} we give the PIT algorithm for $\F\{X\}$. In Section
\ref{sec:fact} we describe the factorization algorithm for
$\F\{X\}$. Finally, we list some open problems in Section
\ref{conclusion}.

\section{Preliminaries}\label{prelim}

For an arithmetic circuit $C$, a \emph{parse tree} for a monomial $m$
is a multiplicative sub-circuit of $C$ rooted at the output gate
defined by the following process starting from the output gate:

\begin{itemize}
\item At each $+$ gate retain exactly one of its input gates.
\item At each $\times$ gate retain both its input gates.
\item Retain all inputs that are reached by this process.
\item The resulting subcircuit is multiplicative and computes a
  monomial $m$ (with some coefficient).
\end{itemize}

For arithmetic circuits $C$ computing polynomials in the free
nonassociative noncommutative ring $\F\{X\}$, the same definition for
the parse tree of a monomial applies. As explained in the
introduction, in this case each parse tree (generating some monomial)
comes with a bracketed structure for the multiplication. It is
convenient to consider a polynomial in $\F\{x_1,\ldots,x_n\}$ as an
element in the noncommutative ring $\F\angle{x_1,\ldots,x_n, (, )}$
where we introduce two auxiliary variables $($ and $)$ (for left and
right bracketing) to encode the parse tree structure of any
monomial. We illustrate the encoding by the following example.
 

Consider the monomial (which is essentially a binary tree with leaves
labeled by variables) in the nonassociative ring $\F\{x,y\}$ shown in
Figure \ref{fig1:sub1}. Its encoding as a bracketed string in the free
noncommutative ring $\F\angle{x,y,(,)}$ is $(( \ x \ y \ ) \ y \ )$
and its parse tree shown in Figure \ref{fig1:sub2}.

\begin{figure}[h]
\label{fig1}
\centering
\begin{subfigure}{.5\textwidth}
  \centering
  \begin{tikzpicture}[scale=0.75]
\node[circle,draw](z){$\times$}
  child{node[circle,draw]{ $\times$}  child{node[]{$x$}} child{node[]{$y$}}}
  child{
    node[]{$y$}};
\end{tikzpicture}
\caption{A nonassociative and noncommutative monomial $x y y$}
  \label{fig1:sub1}
\end{subfigure}%
\begin{subfigure}{.5\textwidth}
  \centering
 \begin{tikzpicture}[scale=0.75,sibling distance=3cm, 
level 2/.style={sibling distance =2cm}, 
level 3/.style={sibling distance =1.6cm},
level 4/.style={sibling distance =1cm},
level 5/.style={sibling distance =.8cm}]
\node[circle,draw](z){\tiny $\times $}
child{node[circle,draw]{\tiny $\times$}
child{node[]{$($}}
child{node[circle,draw]{\tiny $\times$}
child{node[circle,draw]{\tiny $\times$} 
child{node[]{$($}} child{node[]{$x$}} }  
child{node[circle,draw]{\tiny $\times$} 
child{node[]{$y$}} child{node[]{$)$}} }}}
child{node[circle,draw]{\tiny $\times$} 
child{node[]{$y$}} 
child{node[]{$)$}}} ;
\end{tikzpicture}
\caption{ Corresponding monomial $((xy)\ y)\in\F\angle{X}$.}
  \label{fig1:sub2}
\end{subfigure}
\caption{nonassociative \& noncommutative monomial and its corresponding noncommutative bracketed monomial}
\label{fig:test}
\end{figure}

Consider an arithmetic circuit $C$ computing a polynomial
$f\in\F\{X\}$. The circuit $C$ can be efficiently transformed to a
circuit $\tilde{C}$ that computes the corresponding polynomial
$\tilde{f}\in\F\angle{X, (, )}$ by simply introducing the bracketing
structure for each multiplication gate of $C$ in a bottom-up manner as
indicated in the following example figures. Consider the circuits
described in Figures \ref{fig2:sub1} and \ref{fig2:sub2} where
$f_i,g_i,h_i$'s are polynomials computed by subcircuits. Clearly the
bracket variables preserve the parse tree structure. The following
fact is immediate.

\begin{figure}[h]
\label{fig2}
\centering
\begin{subfigure}{.5\textwidth}
  \centering
  \begin{tikzpicture}[sibling distance=1.75cm, level 2/.style={sibling distance =0.75cm}]
\node[circle,draw](z){\small $+$}
 child{node[circle,draw]{\small $\times$} child{node[]{$f_1$}} child{node[]{$f_2$}} } 
 child{node[circle,draw]{\small $\times$} child{node[]{$g_1$}} child{node[]{$g_2$}} }
 child{node[circle,draw]{\small $\times$} child{node[]{$h_1$}} child{node[]{$h_2$}} };
\end{tikzpicture}
\caption{ $C$ computing a nonassociative, \\ noncommutative polynomial.}
  \label{fig2:sub1}
\end{subfigure}%
\begin{subfigure}{.5\textwidth}
  \centering
  \begin{tikzpicture}[sibling distance=2.75cm, level 2/.style={sibling distance =1.25cm}, level 3/.style={sibling distance =.65cm}]
\node[circle,draw](z){\small $+$}

 child{node[circle,draw]{\small $\times$} child{node[circle,draw]{\small $\times$} child{node[]{$($}} child{node[]{$f_1$}}} child{node[circle,draw]{\small $\times$} child{node[]{$f_2$} } child{node[]{$)$} }    }}
 child{node[circle,draw]{\small $\times$} child{node[circle,draw]{\small $\times$} child{node[]{$($}} child{node[]{$g_1$}}} child{node[circle,draw]{\small $\times$} child{node[]{$g_2$} } child{node[]{$)$} }    }}
 child{node[circle,draw]{\small $\times$} child{node[circle,draw]{\small $\times$} child{node[]{$($}} child{node[]{$h_1$}}} child{node[circle,draw]{\small $\times$} child{node[]{$h_2$} } child{node[]{$)$} }    }}
 ;
\end{tikzpicture}
\caption{ $\tilde{C}$ that computes the corresponding noncommutative polynomial.}
  \label{fig2:sub2}
\end{subfigure}
\caption{Nonassociative circuit and its corresponding noncommutative bracketed circuit}
\label{fig:test}
\end{figure}

\begin{obs}\label{obs1}
A nonassociative noncommutative circuit $C$ computes a nonzero
polynomial $f\in\F\{X\}$ if and only if the corresponding 
noncommutative circuit $\tilde{C}$ computes a nonzero polynomial
$\tilde{f}\in\F\angle{X,(, )}$.
\end{obs}

We recall that the free noncommutative ring $\F\angle{\ex}$ is not a
unique factorization domain (UFD) \cite{Cohn} as shown by the
following standard example : $x y x + x = x (y x + 1) = (x y + 1)
x$. In contrast, the nonassociative free ring $\F\{\ex\}$ is a
UFD. 

\begin{proposition}\label{ufd}
  Over any field $\F$, the ring $\F\{\ex\}$ is a unique factorization
  domain. More precisely, any polynomial $f\in\F\{X\}$ can be
  expressed a product $f=g_1g_2\cdots g_r$ of irreducible polynomials
  $g_i\in\F\{X\}$. The factorization is unique upto constant factors
  and reordering.
\end{proposition}

\begin{remark}
Usually, even the ordering of the irreducible factors in the
factorization is unique. Exceptions arise because of the equality
$(g+\alpha)(g+\beta)=(g+\beta)(g+\alpha)$ for any polynomial
$g\in\F\{X\}$ and $\alpha,\beta\in\F$.
\end{remark}

We shall indirectly see a proof of this proposition in
Section~\ref{sec:fact} where we describe the algorithm for computing
all irreducible factors.

Given a noncommutative circuit $C$ computing a homogeneous polynomial
in $\F\angle{\ex}$ and a monomial $m$ over $\ex$, one can talk of the
left and right derivatives of $C$ w.r.t $m$ \cite{AJR15}. Let
$f=\sum_{m'}c_{m'}(f) m'$ for some $f\in\F\angle{\ex}$ and $A$ be the
subset of monomials $m'$ of $f$ that have $m$ as prefix. Then the left
derivative of $f$ w.r.t.\ $m$ is
\[
\frac{\partial^{\ell}f}{\partial m}=\sum_{m'\in A} c_{m'}(f) m'',
\]
where $m'=m\cdot m''$ for $m'\in A$. Similarly we can define the right
derivative $\frac{\partial^{r}f}{\partial m}$. As shown in
\cite{AJR15}, if $f$ is given by a circuit $C$ then in deterministic
polynomial time we can compute circuits for
$\frac{\partial^{\ell}f}{\partial m}$ and
$\frac{\partial^{r}f}{\partial m}$. We briefly discuss this in the
following lemma.

\begin{lemma}{\rm\cite{AJR15}}\label{derivative}
  Given a noncommutative circuit $C$ of size $s$ computing a
  homogeneous polynomial $f$ of degree $d$ in $\F\angle{\ex}$ and
  monomial $m$, there is a deterministic $\poly(n,d,s)$ time algorithm
  that computes circuits $C_{m,\ell}$ and $C_{m,r}$ for the left and
  right derivatives $\frac{\partial^\ell C}{\partial m}$ and
  $\frac{\partial^r C}{\partial m}$, respectively. 
\end{lemma}

\begin{proof}
  We explain only the left partial derivative case. Let $m$ be a
  degree $d'$ monomial and $f\in\F\angle{X}$ be a homogeneous degree
  $d$ polynomial $f$ computed by circuit $C$. In \cite{AJR15},a small
  substitution deterministic finite automaton $A$ with $d'+2$ states
  is constructed that recognizes all length $d$ strings with prefix
  $m$ and substitutes 1 for prefix $m$. The transition matrices of
  this automaton can be represented by $(d'+2)\times(d'+2)$
  matrices. From the evaluation of circuit $C$ on these transition
  matrices will recover the circuit for $\frac{\partial^\ell
    C}{\partial m}$ in the $(1,d^{'} + 1)^{th}$ entry of the output
  matrix.
\end{proof}

The left and right partial derivatives of inhomogeneous polynomials
are similarly defined. The same matrix substitution works for
non-homogeneous polynomials as well \cite{AJR15}. As discussed above,
given a nonassociative arithmetic circuit $C$ computing a polynomial
$f\in\F\{\ex\}$, we can transform $C$ into a noncommutative circuit
$\tilde{C}$ that computes a polynomial $\tilde{f}\in\F\angle{X,(,)}$.
Suppose we want to compute the left partial derivative of $f$
w.r.t.\ a monomial $m\in\F\{X\}$. Using the tree structure of $m$ we
transform it into a monomial $\tilde{m}\in \F\angle{X,(,)}$ and then
we can apply Lemma~\ref{derivative} to $\tilde{C}$ and $\tilde{m}$ to
compute the required left partial derivative. We can similarly compute
the right partial derivative. We use this in Section
\ref{sec:fact}. 

We also note the following simple fact that the homogeneous parts of a
polynomial $f\in\F\{X\}$ given by a circuit $C$ can be computed
efficiently. We can apply the above transformation to obtain
circuit $\tilde{C}$ and use a standard lemma (see e.g., \cite{SY10})
to compute the homogeneous parts of $\tilde{C}$.

\begin{lemma}\label{hom}
  Given a noncommutative circuit $C$ of size $s$ computing a
  noncommutative polynomial $f$ of degree $d$ in $\F\angle{\ex,(,)}$,
  one can compute homogeneous circuits $C_j$ (where each gate computes
  a homogeneous polynomial) for $j^{th}$ homogeneous part $f_j$ of
  $f$, where $0\leq j\leq d$, deterministically in time
  $\poly(n,d,s)$.
\end{lemma}

\section{Identity Testing in $\F\{X\}$}\label{sec:pit}

In this section we describe our identity testing algorithm. 

\begin{theorem}\label{thm-pit}
  Let $f(x_1,x_2,\ldots,x_n)\in\F\{X\}$ be a degree $d$ polynomial
  given by an arithmetic circuit of size $s$. Then in deterministic
  $\poly(s,n,d)$ time we can test if $f$ is an identically zero
  polynomial in $\F\{X\}$.
\end{theorem}

\begin{proof}
By Lemma \ref{hom} we can assume that the input is a homogeneous
nonassociative circuit $C$ computing some homogeneous degree $d$
polynomial in $\F\{X\}$ (i.e.\ every gate in $C$ computes a
homogeneous polynomial).  Also, all the $\times$ gates in $C$ have
fanin 2 and $+$ gates have unbounded fanin. We can assume the output
gate is a $+$ gate. We can also assume w.l.o.g.\ that the $+$ and
$\times$ gates alternate in each input gate to output gate path in the
circuit (otherwise we introduce sum gates with fan-in 1).


The $j^{th}$-layer of circuit $C$ to be the set of all $+$ gates in
computing degree $j$ homogeneous polynomials. Let $s^{+}$ be the total
number of $+$ gates in $C$. To each monomial $m$ we can associate a
vector $v_m\in\F^{s^{+}}$ of coefficients, where $v_m$ is indexed by
the $+$ gates in $C$, and $v_m[g]$ is the coefficient of monomial $m$
in the polynomial computed at the $+$ gate $g$. We can also write
\[
v_m[g] = c_m(p_g),
\]
where $p_g$ is the polynomial computed at the sum gate $g$.

For the $j^{th}$ layer of $+$ gates, we will maintain a maximal
linearly independent set $\mathbb{B}_j$ of vectors $v_m$ of monomials.
These vectors correspond to degree $j$ monomials. Although
$v_m\in\F^{s^{+}}$, notice that $v_m[g]=0$ at all $+$ gates that do
not compute a degree $j$ polynomial. Thus, $|\mathbb{B_j}|$ is bounded
by the number of $+$ gates in the $j^{th}$ layer. Hence,
$|\mathbb{B_j}|\le s$.

The sets $\mathbb{B}_j$ are computed inductively for increasing values
of $j$. For the base case, the set $\mathbb{B}_1$ can be easily
constructed by direct computation.  Inductively, suppose the sets
$\mathbb{B}_i : 1\leq i\leq j-1$ are already constructed. We describe
the construction of $\mathbb{B}_j$. Computing $\mathbb{B}_d$ and checking
if there is a nonzero vector in it yields the identity testing algorithm.

We now describe the construction for the $j^{th}$ layer assuming we
have basis $\B_{j'}$ for every $j' < j$. Consider a $\times$ gate with
its children computing homogeneous polynomials of degree $d_1$ and
$d_2$ respectively. Notice that $j=d_1 + d_2$ and $0 < d_1 , d_2 <
j$. Consider the monomial\footnote{We note that the nonassociative
  monomial $m_1 m_2$ is a binary tree with the root having two
  children: the left child is the root of the binary tree for $m_1$
  and the right child is the root of the binary tree for $m_2$.} set

\[
M =\{ m_1 m_2 \ | \ v_{m_1} \in \mathbb{B}_{d_1} \text{ and } v_{m_2}
\in \mathbb{B}_{d_2} \}.
\]

We construct vectors $\{v_{m} \ | \ m \in M \}$ as follows.

\[
v_{m_1 m_2}[g] = \sum_{(g_{d_1}, g_{d_2})} v_{m_1}[g_{d_1}] v_{m_2}[g_{d_2}],
\]

where $g$ is a $+$ gate in the $j^{th}$ layer, $g_{d_1}$ is a $+$ gate
in the $d_1^{th}$ layer, $g_{d_2}$ is a $+$ gate in the $d_2^{th}$
layer, and there is a $\times$ gate which is input to $g$ and computes
the product of $g_{d_1}$ and $g_{d_2}$.

Let $\mathbb{B}_{d_1 , d_2}$ denote a maximal linearly independent
subset of $\{v_m \mid m \in M \}$.  Then we let $\mathbb{B}_d$ be a
maximal linearly independent subset of
\[
\bigcup_{ d_1 + d_2 = d}\mathbb{B}_{d_1 , d_2}. 
\]

\begin{claim}
  For every monomial $m$ of degree $j$, $v_m$ is in the span of
  $\mathbb{B}_j$.
\end{claim}

\noindent\textit{Proof of Claim.}~~  Let $m=m_1 m_2$ and the degree of
$m_1$ is $d_1$ and the degree of $m_2$ is $d_2$ \footnote{Here a
  crucial point is that for a nonassociative monomial of degree $d$,
  such a choice for $d_1$ and $d_2$ is \emph{unique}. This is a place
  where a general noncommutative circuit behaves very differently.}.
By \emph{Induction Hypothesis} vectors $v_{m_1}$ and $v_{m_2}$ are in
the span of $\mathbb{B}_{d_1}$ and $\mathbb{B}_{d_2}$ respectively.
Hence, we can write
\[
v_{m_1} = \sum^{D_1}_{i=1} \alpha_i v_{m_i} \ \ v_{m_i} \in \mathbb{B}_{d_1}
\textrm{~~~and~~~}
v_{m_2} = \sum^{D_2}_{j=1} \beta_j v_{{m'}_j} \ \ v_{{m'}_j} \in
\mathbb{B}_{d_2},
\]
where $|\mathbb{B}_{d_j}|=D_j$. Now, for a gate $g$ in the $j^{th}$
layer, By Induction Hypothesis and by construction we have
 \begin{align*}
  v_{m}[g] &= \sum_{(g_{d_1} , g_{d_2})} v_{m_1}[g_{d_1}] v_{m_2}[g_{d_2}] 
  = \sum_{g_{d_1} , g_{d_2}} (\sum^{D_1}_{i=1} \alpha_i v_{m_i}[g_{d_1}]) (\sum^{D_2}_{j=1} \beta_j v_{{m'}_j}[g_{d_2}])\\
  &= \sum^{D_1}_{i=1} \sum^{D_2}_{j=1} \alpha_i \beta_j \sum_{g_{d_1} , g_{d_2}} v_{m_i}[g_{d_1}] v_{{m'}_j}[g_{d_2}]
  = \sum^{D_1}_{i=1} \sum^{D_2}_{j=1} \alpha_i \beta_j v_{m_i {m'}_j}[g]. 
 \end{align*}

 Thus $v_m$ is in the span of $\mathbb{B}_{d_1 ,d_2}$ and hence in the
 span of $\mathbb{B}_j$. This proves the claim.

The PIT algorithm only has to check if $\mathbb{B}_d$ has a nonzero
vector. This proves the claim.

Suppose the input nonassociative circuit $C$ computing some degree $d$
polynomial $f\in\F\{X\}$ is inhomogeneous. Then, using Lemma \ref{hom}
we can first compute in polynomial time homogeneous circuits $C_j
~:~0\leq j\leq d$, where $C_j$ computes the degree-$j$ homogeneous
part $f_j$. Then we run the above algorithm on each $C_j$ to check
whether $f$ is identically zero. This completes the proof of the
theorem.
\end{proof}

\section{Polynomial Factorization in $\F\{X\}$}\label{sec:fact}

In this section we describe our polynomial-time white-box
factorization algorithm for polynomials in $\F\{X\}$. More precisely,
given as input a nonassociative circuit $C$ computing a polynomial
$f\in\F\{X\}$, the algorithm outputs circuits for all irreducible
factors of $f$. The algorithm uses as subroutine the PIT algorithm for
polynomial in $\F\{X\}$ described in Section~\ref{sec:pit}.

To facilitate exposition, we completely describe a deterministic
polynomial-time algorithm that computes a nontrivial factorization
$f=g\cdot h$ of $f$, by giving circuits for $g$ and $h$, unless $f$ is
irreducible. We will briefly outline how this extends to finding all
irreducible factors efficiently.

We start with a special case.

\begin{lemma}\label{lem2}
 Let $f\in\F\{\ex\}$ be a degree $d$ polynomial given by a circuit $C$
 of size $s$ such that the constant term in $f$ is zero. Furthermore,
 suppose there is a factorization $f=g\cdot h$ such that the constant
 terms in $g$ and $h$ are also zero. Then in deterministic
 $\poly(n,d,s)$ time we can compute the circuits for polynomials $g$
 and $h$.
\end{lemma}

\begin{proof}
We first consider the even more restricted case when $C$ computes a
homogeneous degree $d$ polynomial $f\in\F\{X\}$. For the purpose of
computing partial derivatives, it is convenient to transform $C$ into
the noncommutative circuit $\tilde{C}$, as explained in
Section~\ref{prelim}, which computes the fully bracketed polynomial
$\tilde{f}\in\F\angle{X,(,)}$. Using Theorem~\ref{thm-pit} we compute
a monomial $m=(m_1 m_2)$ where $m_1$ and $m_2$ are also fully
bracketed. We can transform $\tilde{C}$ to drop the outermost opening
and closing brackets. Now, using Lemma~\ref{derivative}, we compute
the resulting circuits left partial derivative w.r.t.\ $m_1$ and right
partial derivative w.r.t.\ $m_2$. Call these $\tilde{f}_1$ and
$\tilde{f}_2$. We can check if $\tilde{f}=(\tilde{f}_1 \tilde{f}_2)$:
we first recover the corresponding nonassociative circuits for $f_1$
and $f_2$ from the circuits for $\tilde{f}_1$ and $\tilde{f}_2$.  Then
we can apply the PIT algorithm of Theorem~\ref{thm-pit} to check if
$f=f_1 f_2$. Clearly, $f$ is irreducible iff $f\ne f_1
f_2$. Continuing thus, we can fully factorize $f$ into its
irreducible factors.

Now we prove the actual statement. Applying Lemma \ref{hom}, we
compute homogeneous circuits $C_j : 1\leq j\leq d$ for the homogeneous
degree $j$ component $f_j$ of the polynomial $f$.  Clearly $f_d =
g_{d_1} h_{d_2}$. We run the PIT algorithm of Theorem~\ref{thm-pit} on
the circuit $C_d$ to extract a monomial $m$ of degree $d$ along with
its coefficient $c_m(f_d)$ in $f_d$. Notice that the monomial $m$ is
of the form $m=(m_1~m_2)$. If $g$ and $h$ are nontrivial factors of
$f$ then $m_1$ and $m_2$ are monomials in $g$ and $h$
respectively. Compute the circuits for the left and right derivatives
with respect to $m_1$ and $m_2$.
\[
\frac{\partial^{\ell} C_d}{\partial m_{1}}  =  c_{m_1}(g_{d_1}) 
\cdot h_{d_2} \textrm{~~~and~~~}
\frac{\partial^{r} C_d}{\partial m_{2}}  =  c_{m_2}(h_{d_2})\cdot g_{d_1}.
\]

In general the $(i+ d_2)^{th} : i\leq d-d_2$ homogeneous part of $f$
can be expressed as
\[
 f_{i + d_2} = g_i h_{d_2} + \sum_{t=i+1}^{i+ d_2 - 1} g_t h_{{d_2} \ - \ (t - i)}.
\]
We depict the circuit $C_{i+d_2}$ for the polynomial $f_{i+d_2}$ in
Figure \ref{fig4}. The top gate of the circuit is a $+$ gate. From
$C_{i+d_2}$, we construct another circuit $C'_{i+d_2}$ keeping only
those $\times$ gates as children whose left degree is $i$ and right
degree is $d_2$. The resulting circuit is shown in Figure \ref{fig5}.
The circuit $C'_{i+d_2}$ must compute $g_i h_{d_2}$.  By taking the
right partial of $C'_{i+d_2}$ with respect to $m_2$, we obtain the
circuit for $c_{m_2}(h_{d_2})~ g_i$.

\tikzset{
itria/.style={
  draw,dashed,shape border uses incircle,
  isosceles triangle,shape border rotate=90,yshift=-1.45cm},
rtria/.style={
  draw,dashed,shape border uses incircle,
  isosceles triangle,isosceles triangle apex angle=90,
  shape border rotate=-45,yshift=0.2cm,xshift=0.5cm},
ritria/.style={
  draw,dashed,shape border uses incircle,
  isosceles triangle,isosceles triangle apex angle=110,
  shape border rotate=-55,yshift=0.1cm},
letria/.style={
  draw,dashed,shape border uses incircle,
  isosceles triangle,isosceles triangle apex angle=110,
  shape border rotate=235,yshift=0.1cm}
}

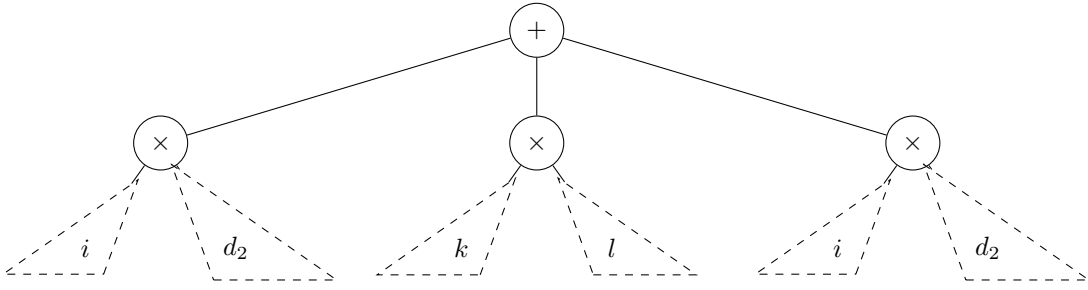
\begin{figure}[h]\label{fig2}
\centering
\begin{tikzpicture}[sibling distance=5cm, level 2/.style={sibling distance =2cm}]
\node[circle,draw](z){\small $+$}
  child{node[circle,draw]{\small $\times$}  child{node[ritria]{\small $i$}} child{node[letria]{\small $d_2$}}}
  child{node[circle,draw]{\small $\times$}  child{node[ritria]{\small $k$}} child{node[letria]{\small $l$}}}
  child{node[circle,draw]{\small $\times$}  child{node[ritria]{\small $i$}} child{node[letria]{\small $d_2$}}};
\end{tikzpicture}
\caption{Circuit $C_{i+d_2}$ for $f_{i+d_2}$}\label{fig4}
\end{figure}

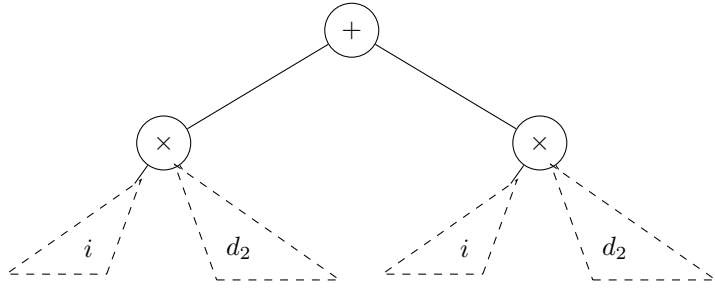
\begin{figure}[h]\label{fig3}
\centering
\begin{tikzpicture}[sibling distance=5cm, level 2/.style={sibling distance =2cm}]
\node[circle,draw](z){\small $+$} 
child{node[circle,draw]{\small $\times$} child{node[ritria]{\small $i$}}
  child{node[letria]{\small $d_2$}}} 
child{node[circle,draw]{\small $\times$} 
child{node[ritria]{\small $i$}} child{node[letria]{\small $d_2$}}};
\end{tikzpicture}
\caption{$C'_{i+d_2}$ keeps only degree $(i,d_2)$ type $\times$ 
gates.}\label{fig5}
\end{figure}

We repeat the above construction for each $i\in[d_1]$ to obtain
circuits for $c_{m_2}(h_{d_2}) g_{i}$ for $1 \leq i \leq
d_1$. Similarly we can get the circuits for $c_{m_1}(g_{d_1}) h_i$ for
each $i\in [d_2]$ using the left derivatives with respect to the
monomial $m_1$.

By adding the above circuits we get the circuits $C_g$ and $C_h$ for
$c_{m_2}(h_{d_2}) g$ and $c_{m_1}(g_{d_1}) h$ respectively.  We set
$C_g = \frac{c_{m_2}(h_{d_2})}{c_{m}(f)} g$ so that $C_g C_h =f$.
Using PIT algorithm one can easily check whether $g$ and $h$ are
nontrivial factors. In that case we further recurse on $g$ and $h$ to
obtain their irreducible factors.
\end{proof}

Now we consider the general case when $f$ and its factors $g, h$ have
arbitrary constant terms. In the subsequent proofs we assume, for
convenience, that $\deg(g)\geq \deg(h)$. The case when $\deg(g) <
\deg(h)$ can be handled analogously.  We first consider the case
$\deg(g)=\deg(h)$.

\begin{lemma}\label{lem3a}
  For a degree $d$ polynomial $f\in \F\{X\}$ given by a circuit $C$
  suppose $f=(g+\alpha)(h+\beta)$, where $g,h\in\F\{X\}$ such that
  $\deg(g)=\deg(h)$, and $\alpha,\beta\in\F$. Suppose $m=(m_1 m_2)$ is
  a nonzero degree $d$ monomial. Then, in deterministic polynomial
  time we can compute circuits for the polynomials $c_{m_1}(g)\cdot h$
  and $c_{m_2}(h)\cdot g$, where $c_{m_1}(g)$ and $c_{m_2}(h)$ are
  coefficient of $m_1$ and $m_2$ in $g$ and $h$ respectively.
\end{lemma}

\begin{proof}
  We can write $f=(g+\alpha)(h+\beta) = g\cdot h + \beta\cdot g +
  \alpha\cdot h + \alpha\cdot \beta$. Applying the PIT algorithm of
  Theorem~\ref{thm-pit} on $f$, we compute a maximum degree monomial
  $m=(m_1 m_2)$. Computing the left derivative of circuit $C$ w.r.t.\
  monomial $m_1$, after removing the outermost brackets, we obtain a
  circuit computing $c_{m_1}(g) h + \beta c_{m_1}(g) + \alpha
  c_{m_1}(h)$. Dropping the constant term, we obtain a circuit
  computing polynomial $c_{m_1}(g)h$.  Similarly, computing the right
  derivative w.r.t $m_2$ yields a circuit for $c_{m_2}(h) g + \beta
  c_{m_2}(g) + \alpha c_{m_2}(h)$. Removing the constant term we get a
  circuit for $c_{m_2}(h)g$.
\end{proof}

When $\deg(g) > \deg(h)$ we can recover $h+ \beta$ entirely (upto a
scalar factor) and we need to obtain the homogeneous parts of $g$
separately.

\begin{lemma}\label{lem3}
  Let $f=(g+ \alpha)\cdot (h+ \beta)$ be a polynomial of degree $d$ in
  $\F\{\ex\}$ given by a circuit $C$. Suppose $\deg(g) >
  \deg(h)$. Then, in deterministic polynomial time we can compute the
  circuit $C'$ for $c_{m_1}(g)(h+\beta)$.
\end{lemma}

\begin{proof}
  Again, applying the PIT algorithm to $f$ we obtain a nonzero degree
  $d$ monomial $m=(m_1~m_2)$ of $f$. If $f=(g+\alpha)(h+\beta)$ then
  $f=g\cdot h + \alpha h + \beta g + \alpha\beta$. As
  $\deg(g)>\deg(h)$, the left partial derivative of $C$ with respect
  to $m_1$ yields a circuit $C'$ for $c_{m_1}(g)~(h + \beta)$.
\end{proof}

Extracting the homogeneous components from the circuit $C'$ given by
Lemma~\ref{lem3}, yields circuits for $\{ c_{m_1}(g) h_i \ : \ i \in
[d_2] \}$. We also get the constant term $c_{m_1}(g) \beta$.  Now we
obtain the homogeneous components of $g$ as follows.

\begin{lemma}\label{lem4a}
  Suppose circuit $C$ computes $f$, where $f=(g+\alpha)~(h+\beta)$ of
  degree $d$, $\alpha,\beta\in\F$, $\deg(g)=d_1$ and $\deg(h)=d_2$
  such that $d_1 > d_2$. 
\begin{itemize}
\item Let $m$ be a nonzero degree $d$ monomial of $f$ such that
  $m=(m_1~m_2)$. Then circuits for $\{ c_{m_2}(h) g_i : i \in [d_1 -
  d_2 +1, d_1] \}$ can be computed in deterministic polynomial time.

\item The $(d_2 + i)^{th}$ homogeneous part of $f$ is given by $f_{d_2
    + i}=\sum_{j=0}^{d_2-1}~g_{d_2 + i -j}~h_{j} + g_{i} ~h_{d_2}$ for
  $1\leq i\leq d_1 - d_2$.  From the circuit $C_{d_2 + i}$ of $f_{d_2
    + i}$, we can efficiently compute circuits for $\{ c_{m_2}
  (h_{d_2}) g_{i} : 1\leq i\leq d_1 - d_2\}$.
\end{itemize}

\end{lemma}

\begin{proof}
  For the first part, fix any $i\in[d_1-d_2+1, d_1]$, and compute the
  homogeneous $(i + d_2)^{th}$ part $f_{i+d_2}$ of $f$ by a circuit
  $C_{i + d_2}$. Similar to Lemma \ref{lem2}, we focus on the
  sub-circuits of $C_{i + d_2}$ formed by $\times$ gate of the degree
  type $(i, d_2)$. Since $i$ is at least $d_1 - d_2 + 1$, such gates
  can compute the multiplication of a degree $i$ polynomial with a
  degree $d_2$ polynomial.  Then, by taking the right partial
  derivative with respect to $m_2$ we recover the circuits for
  $c_{m_2}(h_{d_2}) ~g_i$ for any $i\in[d_1-d_2+1, d_1]$.
 
  Next, the goal is to recover the circuits for $g_i$ (upto a scalar
  multiple), where $1\leq i\leq d_1-d_2$, and also recover the
  constant terms $\alpha$ and $\beta$. When $i \leq d_1 - d_2$ a
  product gate of type $(i,d_2)$ can entirely come from $g$ which
  requires a different handling.

  We explain only the case when $i=d_1 - d_2$ (the others are
  similar). For $i=d_1-d_2$, we have $f_{d_1} = \beta g_{d_1} +
  \sum_{j=1}^{d_2-1}~g_{d_1 -j}~h_{j} + g_{d_1 - d_2} ~h_{d_2}$. By
  Lemma \ref{lem3}, we can compute a circuit $C'$ for $c_{m_1}(g) (h +
  \beta)$. Extracting the constant term yields $c_{m_1}(g)
  \beta$. From Lemma \ref{lem4a} we have a circuit $C''$ for
  $c_{m_2}(h) g_{d_1}$. Multiplying these circuits, we obtain a
  circuit $C^*$ for $c_{m_2}(h)c_{m_1}(g)\beta g_{d_1}$.  Since
  $c_{m_2}(h)c_{m_1}(g) = c_{m}(f)$, dividing $C^*$ by $c_m (f)$
  yields a circuit for $\beta g_{d_1}$. Note that, by the first part
  of this lemma, we already have circuits for every term $g_{d_1-j}$
  appearing in the above sum. Subtracting $\beta g_{d_1} +
  \sum_{j=1}^{d_2-1}~g_{d_1 -j}~h_{j}$ from the circuit $C_{d_1}$ for
  $f_{d_1}$, yields a circuit for polynomial $g_{d_1 - d_2} h_{d_2}$.
  Computing the right derivative of the resulting circuit w.r.t $m_2$
  (Lemma~\ref{derivative}) yields a circuit for $c_{m_2}(h) g_{d_1 -
    d_2}$.
 
  For general $i\le d_1-d_2$, when we need to compute $g_i$, again we
  will have already computed circuits for all $g_j, j>i$. A suitable
  right derivative computation will yield a circuit for $c_{m_2}(h)
  g_{i}$.
\end{proof}

Lemmas \ref{lem2}, \ref{lem3a}, \ref{lem3}, and \ref{lem4a} yield an
efficient algorithm for computing circuits for the two factors
$c_{m_2}(h)(\sum_{i=1}^{d_1} g_i)$ and $c_{m_1}(g) (\sum_{i=1}^{d_2}
h_i)$ when $\deg(g) \ge \deg(h)$. The case when $\deg(g) < \deg(h)$ is
similarly handled using left partial derivatives in the above lemmas.

Now we explain how to compute the constant terms of the individual
factors. We discuss the case when $\alpha\ne 0$. The other case is
similar.

First we recall that given a monomial $m$ and a noncommutative circuit
$C$, the coefficient of $m$ in $C$ can be computed in deterministic
polynomial time \cite{AMS10}.  We know that $f_0 = \alpha\cdot \beta$.
We compute the coefficient of the monomial $m_1$ in the circuits for
polynomials $c_{m_2}(h)c_{m_1}(g) g h$, $c_{m_2}(h) g$, and
$c_{m_1}(g)h$. Let these coefficients be $a, b$ and $c$,
respectively. Moreover, we know that $c_{m_2}(h)c_{m_1}(g)$ is the
coefficient of monomial $m=(m_1 ~m_2)$ in $f$. Let the coefficient of
$m_1$ in $f$ be $\gamma$. Let $\gamma_1=c_{m_1}(g)$ and
$\gamma_2=c_{m_2}(h)$ and $\delta=c_{m_1}(g)c_{m_2}(h)$.

Now equating the coefficient of $m_1$ from both side of the equation
$f = (g + \alpha) (h + \beta)$ and substituting
$\beta=\frac{f_0}{\alpha}$, we get
\[
\gamma = \frac{a}{\gamma_1\gamma_2} + \frac{\alpha c}{\gamma_1} + 
\frac{f_0 b}{\alpha \gamma_2} =  \frac{a}{\delta} + \frac{\alpha c}{\gamma_1} + 
\frac{f_0 b}{\alpha \gamma_2}.
\]

Letting $\xi=\alpha \gamma_2$, this gives a quadratic equation in
the unknown $\xi$.
\[
c\xi^2 + (a - \gamma \delta)\xi + f_0 b\delta = 0.
\]

By solving the above quadratic equation we get two solutions $A_1$ and
$A_2$ for $\xi=\alpha\gamma_2$. Notice that $\beta\gamma_1 =
\frac{\delta f_0}{\xi}$. As we have circuits for $c_{m_2}(h) g=
\gamma_2 g$ and for $c_{m_1}(g) h= \gamma_1 h$, we obtain circuits for
$\gamma_2(g+\alpha)$ and $\gamma_1(h+\beta)$ (two solutions,
corresponding to $A_1$ and $A_2$). To pick the right solution, we can
run the $\pit$ algorithm to check if $\gamma_1\gamma_2 f$ equals the
product of these two circuits that purportedly compute
$\gamma_2(g+\alpha)$ and $\gamma_1(h+\beta)$.

Over $\Q$ we can just solve the quadratic equation in deterministic
polynomial time using standard method. If $\F=\F_q$ for $q=p^r$, we
can factorize the quadratic equation in deterministic time
$\poly(p,r)$ \cite{GS92}. Using randomness, one can solve this problem
in time $\poly(\log p,r)$ using Berlekamp's factoring algorithm
\cite{Ber71}. This also completes the proof of the following.

\begin{theorem}\label{thm-fact1}
Let $f\in\F\{X\}$ be a degree $d$ polynomial given by a circuit of
size $s$. If $\F=\Q$, in deterministic $\poly(s,n,d)$ time we can
compute a nontrivial factorization of $f$ or reports $f$ is
irreducible. If $\F$ is a finite field such that $char(\F)=p$, we
obtain a deterministic $\poly(s,n,d,p)$ time algorithm that computes a
nontrivial factorization of $f$ or reports $f$ is irreducible.
\end{theorem}

Finally, we state the main result of this paper.

\begin{theorem}\label{thm-fact2}
  Let $f\in\F\{X\}$ be a degree $d$ polynomial given by a circuit of
  size $s$. Then if $\F=\Q$, in deterministic $\poly(s,n,d)$ time we
  can output the circuits for the irreducible factors of $f$. If $\F$
  is a finite field such that $char(\F)=p$, we obtain a deterministic
  $\poly(s,n,d,p)$ time algorithm for computing circuits for the
  irreducible factors of $f$.
\end{theorem}

\begin{remark}
We could apply Theorem~\ref{thm-fact1} repeatedly to find all
irreducible factors of the input $f\in\F\{X\}$. However, the problem
with that approach is that the circuits for $g$ and $h$ we computed in
the proof of Theorem~\ref{thm-fact1}, where $f=gh$ is the
factorization, is larger than the input circuit $C$ for $f$ by a
polynomial factor. Thus, repeated application would incur a
superpolynomial blow-up in circuit size. We can avoid that by
computing the required partial derivative of $g$ as a suitable partial
derivative of the circuit $C$ directly. This will keep the circuits
polynomially bounded. This idea is from \cite{AJR15} where it is used
for homogeneous noncommutative polynomial factorization. Combined
with Theorem~\ref{thm-fact1} this gives the polynomial-time algorithm 
of Theorem~\ref{thm-fact2}.
\end{remark}
 
\section{Conclusion}\label{conclusion}

Motivated by the nonassociative circuit lower bound result shown in
\cite{HWY10b}, we study PIT and polynomial factorization in the free
nonassociative noncommutative ring $\F\{X\}$ and obtain efficient
white-box algorithms for the problems.

Hrubes, Wigderson, and Yehudayoff \cite{HWY10b} have also shown
exponential circuit-size lower bounds for nonassociative, commutative
circuits. It would be interesting to obtain an efficient polynomial
identity testing algorithm for that circuit model too. Even a
randomized polynomial-time algorithm is not known.

Obtaining an efficient \emph{black-box} PIT in the ring $\F\{X\}$ is
also an interesting problem. Of course, for such an algorithm the
black-box can be evaluated on a suitable nonassociative algebra. To
the best of our knowledge, there seems to be no algorithmically useful
analogue of the Amitsur-Levitzki theorem \cite{AL50}.

\end{document}